\documentclass [pra, aps, twocolumn, superscriptaddress]
               {revtex4-1}
\usepackage [dvips] {graphicx}
\usepackage {amsmath, amssymb, amsthm, dsfont, fullpage,
  hyperref, color, ctable}
\usepackage [USenglish] {babel}

\newcommand{\bra}[1]{\langle#1|}
\newcommand{\ket}[1]{|#1\rangle}

\newcommand{\Tr}{\operatorname{Tr}}
\newcommand{\rank}{\operatorname{rank}}
\newcommand{\lin}[1]{\mathcal{L}(#1)}
\newcommand{\hilb}[1]{\mathcal{#1}}
\newcommand{\floor}[1]{\lfloor#1\rfloor}
\newcommand{\kl} [2] {\operatorname{D}\left(#1\left|\left|#2\right.\right.\right)}

\newtheorem{dfn}{Definition}
\newtheorem{lmm}{Lemma}
\newtheorem{thm}{Theorem}
\newtheorem{cor}{Corollary}
\newtheorem{rmk}{Remark}

\begin{document}

\title{Hierarchy  of Bounds  on  Accessible Information  and
  Informational Power}

\author{Michele \surname{Dall'Arno}}

\email{cqtmda@nus.edu.sg}

\affiliation{Centre   for  Quantum   Technologies,  National
  University  of  Singapore,  3   Science  Drive  2,  117543
  Singapore, Republic of Singapore}

\date{\today}

\begin{abstract}
  Quantum  theory  imposes  fundamental limitations  to  the
  amount of information  that can be carried  by any quantum
  system.   On the  one  hand, Holevo  bound  rules out  the
  possibility to encode more information in a quantum system
  than in its classical  counterpart, comprised of perfectly
  distinguishable states. On the other hand, when states are
  uniformly distributed  in the  state space,  the so-called
  subentropy lower  bound is saturated. How  uniform quantum
  systems are can be  naturally quantified by characterizing
  them as  $t$-designs, with  $t = \infty$  corresponding to
  the uniform distribution. Here we  show the existence of a
  trade-off between  the uniformity of a  quantum system and
  the amount  of information it  can carry. To this  aim, we
  derive a  hierarchy of informational bounds  as a function
  of $t$ and  prove their tightness for  qubits and qutrits.
  By deriving  asymptotic formulae for large  dimensions, we
  also show that the  statistics generated by any $t$-design
  with  $t >  1$  contains  no more  than  a  single bit  of
  information, and  this amount  decreases with  $t$. Holevo
  and subentropy  bounds are  recovered as  particular cases
  for $t = 1$ and $t = \infty$, respectively.
\end{abstract}

\maketitle

\section{Introduction}
\label{sec:introduction}

Quantum theory imposes fundamental limitations to the amount
of  information  that can  be  {\em  encoded into}  or  {\em
  extracted from}  any quantum system. Formally,  the former
case  is  referred to  as  the  problem of  {\em  accessible
  information}~\cite{LL66,   Hol73,  Bel75a,Bel75b,   Dav78,
  JRW94}  of a  quantum ensemble,  while the  latter as  the
problem  of {\em  informational power}~\cite{DDS11,  OCMB11,
  Hol12,  Hol13,  SS14, DBO14,  Szy14,  Dal14,  Dal15} of  a
quantum  measurement. Recently,  a duality  relation between
these  two  quantities  was  established~\cite{DDS11},  that
allows to generally refer to  the problem of quantifying the
information {\em carried by} a quantum system.

On   the    one   hand,   the   well-known    Holevo   upper
bound~\cite{Hol73} rules out the  possibility to encode more
information in a quantum system than in its purely classical
counterpart,   comprised    of   perfectly   distinguishable
states. On  the other hand,  for a genuinely  quantum system
whose states  are uniformly distributed in  the state space,
the   so-called  subentropy   lower  bound~\cite{JRW94}   is
saturated. Therefore, one might  conjecture the existence of
a  general trade-off  between  the uniformity  of a  quantum
system and the amount of information it can carry.

A natural mean to quantify  how uniform quantum systems are,
is provided by their  characterization in terms of spherical
quantum $t$-designs~\cite{Zau99,  RBSC04, KR05,  AE07, SG10,
  BWB10}, with $t = 1$ corresponding to an arbitrary quantum
measurement and $t = \infty$ corresponding to the completely
uniform   distribution.    Other  remarkable   examples   of
$t$-designs    for   the    case   $t=2$    are   symmetric,
informationally        complete         (SIC)        quantum
measurements~\cite{Zau99,   KR05}  and   complete  sets   of
mutually  unbiased  bases (MUBs)~\cite{AE07,  SG10,  BWB10}.
They play a  fundamental role in a  plethora of applications
such       as       quantum       tomography~\cite{DDPPS02},
cryptography~\cite{BB84},                        information
locking~\cite{DHLST04},      quantumness     of      Hilbert
space~\cite{FS03,Fuc04},         entropic        uncertainty
relations~\cite{San95, BW07, WW10, BR11, BHOW13, Ras13}, and
foundations  of  quantum   theory~\cite{FS09,  FS11,  AEF11,
  Fuc12}.

In this  work, for any $t$  we derive an upper  bound on the
information that can be carried by any quantum $t$-design as
a function of  the dimension of the system.   In this sense,
the resulting hierarchy of  bounds proves the correctness of
the above  mentioned conjecture and formally  quantifies it.
Furthermore, we show the tightness  of our bounds for qubits
and  qutrits.  By  deriving  asymptotic  formulae for  large
dimensions, we  also show  that the statistics  generated by
any $t$-design with  $t > 1$ contains no more  than a single
bit of information, and that this amount decreases with $t$.
The Holevo upper bound~\cite{Hol73} and the subentropy lower
bound~\cite{JRW94} are recovered as  particular cases for $t
= 1$ and $t = \infty$, respectively.

The  paper is  structured as  follows.  First,  we introduce
quantum $t$-designs in  Section~\ref{sec:design}, we discuss
the       relevant      figures       of      merit       in
Section~\ref{sec:information}, and provide a way to estimate
them in Section~\ref{sec:interpolation}.  Then, we introduce
our main result,  namely a hierarchy of upper  bounds on the
accessible  information  and   the  informational  power  of
$t$-designs  in  Section~\ref{sec:main},   we  derive  close
analytic expressions  for low  values of $t$  and asymptotic
formulae        for         large        dimension        in
Section~\ref{sec:application},  and we  prove tightness  for
qubits  and  qutrits   in  Section~\ref{sec:tightness}.   We
conclude  by summarizing  our  results  and presenting  some
outlooks in Section~\ref{sec:conclusion}.

\section{Formalization}
\label{sec:formalization}

\subsection{Spherical quantum $t$-designs}
\label{sec:design}

In this  subsection we  recall some  basic facts~\cite{NC00}
from quantum information theory,  and specialize them to the
case of spherical quantum $t$-design.

Any  quantum  system is  associated  with  an Hilbert  space
$\hilb{H}$, and  we denote  with $L(\hilb{H})$ the  space of
linear  operators  on  $\hilb{H}$.  We  will  only  consider
finite-dimensional Hilbert spaces.

A   {\em  quantum   state}  $\rho$   is  represented   by  a
positive-semidefinite  operator in  $L(\hilb{H})$ such  that
$\Tr[\rho] \le 1$.   A {\em pure state} $\psi$  is such that
$\rank\psi = 1$ and is denoted in Dirac notation by a vector
$\ket{\psi}$  with   $\psi  =   \ket{\psi}\bra{\psi}$.   Any
quantum  preparation is  represented  by  an {\em  ensemble}
$\rho_x$, namely  a measurable function $\rho_x$  from reals
to states  such that $\int_x  \Tr[\rho_x] dx = 1$.   An {\em
  ensemble of pure  states} is such that $\rho_x  \neq 0$ if
and  only if  $\rho_x$ is  a pure  state.  The  {\em uniform
  ensemble} is the ensemble  of pure states distributed with
the uniform (Haar) measure on the unit sphere of $\hilb{H}$.

A   {\em  quantum   effect}  $\pi$   is  represented   by  a
positive-semidefinite  operator in  $L(\hilb{H})$ such  that
$\pi   \le   \openone_d$,    where   $\openone_d$   is   the
$d$-dimensional identity operator.   Any quantum measurement
is represented  by a {\em positive-operator  valued measure}
(POVM) $\pi_y$,  namely a  measurable function  $\pi_y$ from
reals to effects  such that $\int_y \pi_y  dy = \openone_d$.

For  any  ensemble  $\rho_x$  and POVM  $\pi_y$,  the  joint
probability density  $p_{x,y}$ of input $x$  and outcome $y$
is given by the {\em Born rule}, i.e.  $p_{x,y} = \Tr[\rho_x
  \pi_y]$.

\begin{dfn}[Spherical quantum $t$-design]
  \label{def:tdesign}
  A  {\em  spherical  quantum  $t$-design}  is  an  ensemble
  $\rho_x$ such that
  \begin{align*}
    \int \frac { \rho_x^{\otimes k}} { \Tr[ \rho_x ]^{k-1} }
    dx  :=  \int  \frac  {\psi_x^{\otimes k}}  {  ||  \psi_x
      ||^{2(k-1)} } dx
  \end{align*}
  holds for  any $k  \le t$, where  $\psi_x$ is  the uniform
  ensemble.
\end{dfn}

\begin{lmm}
  \label{lmm:average}
  Let $\psi_x$ be the uniform ensemble. Then one has
  \begin{align*}
    \int \frac {\psi_x^{\otimes k}}  { || \psi_x ||^{2(k-1)}
    } dx = \binom{d-1+k}{k}^{-1} P_{\textrm{sym}},
  \end{align*}
  where  $P_{\textrm{sym}}$   is  the  projector   over  the
  symmetric subspace  of $\hilb{H}^{\otimes  k}$ and  $d$ is
  the dimension of $\hilb{H}$.
\end{lmm}

\begin{proof}
  See Ref.~\cite{AE07}.
\end{proof}

\begin{rmk}
  From            Definition~\ref{def:tdesign}           and
  Lemma~\ref{lmm:average}  it immediately  follows that  any
  POVM is a $1$-design up to a normalization factor of $d$.
\end{rmk}

Remarkable   examples    of   $2$-designs    are   symmetric
informationally complete (SIC)  POVMs~\cite{Zau99} and $d+1$
mutually unbiased bases (MUBs)~\cite{KR05}.

A  concept that  will be  relevant in  the following  is the
so-called index of coincidence~\cite{Ras13}.

\begin{dfn}[Index of coincidence]
  \label{def:coincidence}
  For any POVM $\pi_y$ and  any unit-trace state $\rho$, the
  {\em index of coincidence} $C_k(\pi_y, \rho)$ is given by
  \begin{align*}
    C_k(\pi_y, \rho)  := \int  \frac{ \Tr[\pi_y \rho]^k  } {
      \Tr[\pi_y]^{k-1} } dy.
  \end{align*}
\end{dfn}

The following result characterizes  the index of coincidence
of $t$-designs.

\begin{lmm}
  \label{thm:coincidence}
  Let $\hilb{H}$  be a  $d$-dimensional Hilbert  space.  Let
  $\pi_y  \in \lin{\hilb{H}}$  be  a  $t$-design POVM.   Let
  $\ket{\psi} \in \hilb{H}$ be  a unit-trace pure state. For
  any $k \le t$, the index of coincidence $C_k(\pi_y, \psi)$
  is independent of $\pi_y$ and $\psi$ and is given by
  \begin{align*}
    C_k = d \binom{d-1+k}{k}^{-1}.
  \end{align*}
\end{lmm}

\begin{proof}
  By Definition~\ref{def:coincidence} one has
  \begin{align*}
    C_k(\pi_y,  \psi) =  \int  \Tr  \left[ \psi^{\otimes  k}
      \frac  {\pi_y^{\otimes  k}}  { \Tr[  \pi_y  ]^{k-1}  }
      \right] dy.
  \end{align*}
  By Lemma~\ref{lmm:average} one has
  \begin{align*}
    C_k(\pi_y,   \psi)   =  d   \binom{d-1+k}{k}^{-1}   \Tr[
      \psi^{\otimes k} P_{\textrm{sym}} ].
  \end{align*}
  Since   $\psi^{\otimes  k}$   belongs  to   the  symmetric
  subspace, the statement immediately follows.
\end{proof}

\subsection{Informational measures}
\label{sec:information}

In     this    subsection     we    recall     some    basic
definitions~\cite{Cov06} from classical information theory.

Given two  probability densities  $p_x$ and $q_y$,  the {\em
  relative entropy} $\kl{p_x}{q_y}$ given by
\begin{align*}
  \kl{p_x}{q_y} := \int p_x \ln\frac{p_x}{q_x} dx
\end{align*}
is a non-symmetric  measure of the distance  between the two
densities.   Given   two  random   variables  $X$   and  $Y$
distributed according to probability density $p_{x, y}$, the
{\em mutual information} $I(X;Y)$ given by
\begin{align*}
  I(X;Y) := \kl{p_{x,y}}{p_x p_y},
\end{align*}
is a measure of their correlation. For any ensemble $\rho_x$
and  POVM $\pi_y$,  we  denote with  $I(\rho_x, \pi_y)$  the
mutual information $I(X;Y)$ between random variables $X$ and
$Y$ distributed according to $p_{x,y} = \Tr[\rho_x \pi_y]$.

The   accessible   information~\cite{LL66,  Hol73,   Bel75a,
  Bel75b} of  an ensemble measures how  much information can
be extracted from the ensemble.

\begin{dfn}[Accessible information]
  \label{def:accinfo}
  The  {\em   accessible  information}  $A(\rho_x)$   of  an
  ensemble $\rho_x$ is the supremum over any POVM $\pi_y$ of
  the mutual information $I(\rho_x, \pi_y)$, namely
  \begin{align*}
    A(\rho_x) := \sup_{\pi_y} I(\rho_x, \pi_y).
  \end{align*}
\end{dfn}

The informational power~\cite{DDS11} of  a POVM measures how
much information can be extracted by the POVM.

\begin{dfn}[Informational power]
  \label{def:infopower}
  The {\em informational power} $W(\pi_y)$ of a POVM $\pi_y$
  is the supremum  over any ensemble $\rho_x$  of the mutual
  information $I(\rho_x, \pi_y)$, namely
  \begin{align*}
    W(\pi_y) := \sup_{\rho_x} I(\rho_x, \pi_y).
  \end{align*}
\end{dfn}

The following Lemma allows us  to recast upper bounds on the
informational   power  into   bounds   for  the   accessible
information.   Therefore in  the following  without loss  of
generality we focus on the former problem.

\begin{thm}
  \label{thm:duality}
  For any  POVM $\pi_y$, the informational  power $W(\pi_y)$
  is given by
  \begin{align*}
     W(\pi_y) = \sup_{\rho} A(\rho^{1/2} \pi_y \rho^{1/2}).
  \end{align*}
\end{thm}

\begin{proof}
  See Refs.~\cite{DDS11, DBO14}.
\end{proof}

In  the following  it will  be convenient  to introduce  the
shorthand notation $\eta(x) := -x \ln x$.

\begin{thm}
  \label{thm:bound}
  For any $d$-dimensional POVM $\pi_y$ one has
  \begin{align*}
    &  W(\pi_y) \\  \le &  \ln d  - \inf_{\psi_x}  \iint ||
    \psi_x ||^2  \Tr[\pi_y] \eta \left(  \frac {\bra{\psi_x}
      \pi_y \ket{\psi_x}}{||  \psi_x ||^2 \Tr[\pi_y]}\right)
    dx dy
  \end{align*}
  where the  infimum is over  any ensemble $\psi_x$  of pure
  states.
\end{thm}

\begin{proof}
  A  Davies-like  theorem  applies~\cite{DDS11},  so  it  is
  sufficient  to maximize  over ensembles  $\psi_x$ of  pure
  states. For any such ensemble one has
  \begin{align*}
    &  I(\psi_x,  \pi_y)  \\   =  &  \kl{\bra{\psi_x}  \pi_y
      \ket{\psi_x}}{\Tr[\rho \pi_y] || \psi_x  ||^2 } \\ = &
    \iint   \bra{\psi_x}  \pi_y   \ket{\psi_x}  \ln   \frac
          {\bra{\psi_x} \pi_y \ket{\psi_x}}  {|| \psi_x ||^2
            \Tr[\rho \pi_y]} dx dy \\ = & \iint \bra{\psi_x}
          \pi_y   \ket{\psi_x}   \left[    \ln   \frac   {d
              \bra{\psi_x}  \pi_y   \ket{\psi_x}}{||  \psi_x
              ||^2  \Tr[\pi_y]}  -  \ln \frac  {d  \Tr[\rho
                \pi_y]} {\Tr[\pi_y]}  \right] dx  dy \\  = &
          \kl{\bra{\psi_x}  \pi_y   \ket{\psi_x}}  {\frac{||
              \psi_x  ||^2  \Tr[\pi_y]}d}  -  \kl  {\Tr[\rho
              \pi_y]}   {\frac{\Tr[\pi_y]}d}    \\   \le   &
          \kl{\bra{\psi_x}  \pi_y   \ket{\psi_x}}  {\frac{||
              \psi_x ||^2 \Tr[\pi_y]}d},
  \end{align*}
  where the final inequality holds  due to the positivity of
  relative  entropy. Then  the statement  follows by  direct
  inspection.
\end{proof}

\subsection{Polynomial interpolation}
\label{sec:interpolation}

In this  subsection we  introduce an  optimization technique
based  on Hermite  polynomial  interpolation. The  following
Lemma bounds the error made in interpolating a function with
a polynomial.

\begin{lmm}
  \label{thm:error}
  Let $a$  and $b$ be reals  and $t$ be a  positive integer.
  Let $f(x)$ be a  real function with continuous derivatives
  up to order $t + 1$ on $[a,b]$. Let $\{ x_i \}_{i=1}^m$ be
  reals such that  $a \le x_i \le b$ and  $x_i < x_{i'}$ for
  any  $i  <  i'$.   Let $\{  j_i  \}_{i=1}^m$  be  positive
  integers such  that $\sum_i j_i =  t - 1$.  Let  $p(x)$ be
  the polynomial  of degree $t$  that agrees with  $f(x)$ at
  $x_i$ up to  derivative of order $j_i-1$ for $1  \le i \le
  m$, namely
  \begin{align*}
    p^{(j_i)}(x_i) = f^{(j_i)}(x_i), \qquad 0 \le i \le m.
  \end{align*}
  For  any  $x  \in  [a,b]$  there  exists  $x'$  such  that
  $\min(x,x_1) < x' < \max(x,x_m)$ and
  \begin{align*}
    f(x) - p(x) = \frac{f^{(t+1)}(x')}{(t+1)!} \prod_{i=1}^m
    (x - x_i)^{k_i}
  \end{align*}
\end{lmm}

\begin{proof}
  See Ref.~\cite{SB02}.
\end{proof}

The following Lemma, derived in Ref.~\cite{SS14}, provides a
polynomial lower bound to a  function. We reproduce here its
proof for completeness.

\begin{lmm}
  \label{thm:strategy}
  Let $a$  and $b$ be reals  and $t$ be a  positive integer.
  Let $f(x)$ be a  real function with continuous derivatives
  up to  order $t+1$ on  $[a,b]$ such that $f^{(j)}(x)  < 0$
  for even $j$ and $f^{(j)}(x) >  0$ for odd $j$, for any $j
  > 1$    and    $x    \in     [a,b]$.     Let    $\{    x_i
  \}_{i=1}^{\floor{t/2}}$ be reals  such that $a <  x_i < b$
  and $x_i < x_{i'}$ for any $i < i'$. The polynomial $p(x)$
  of degree  $t$ such that $p(a)  = f(a)$, $p(b) =  f(b)$ if
  $t$ is odd, and
  \begin{align*}
    p^{(j)}(x_i) & =  f^{(j)}(x_i), \quad & \forall  1 \le i
    \le \floor{t/2}, \; j = 0,1,
  \end{align*}
  is such that $p(x) \le f(x)$ for $x \in [a,b]$.
\end{lmm}

\begin{proof}
  See Ref.~\cite{SS14}. Let us distinguish two cases. If $t$
  is odd then
  \begin{align*}
    (x-a) (x-b) \prod_{i=1}^{\floor{t/2}} (x-x_i)^2 \le 0 ,
  \end{align*}
  for $x \in [a,b]$. If $t$ is even then
  \begin{align*}
    (x-a) \prod_{i=1}^{\floor{t/2}} (x-x_i)^2 \ge 0,
  \end{align*}
  for $x \in [a,b]$.  Then the statement immediately follows
  from Lemma~\ref{thm:error}.
\end{proof}

\section{Informational bounds}
\label{sec:bound}

\subsection{Main result}
\label{sec:main}

The informational power problem is formally the optimization
of  an  entropic  function  over  complex  vectors  under  a
normalization constraint, therefore it  is unfeasible in the
majority of  cases.  However,  in this subsection  we recast
the informational power problem for $t$-design POVMs into an
unconstrained optimization over $\floor{t/2}$ real variables
$\{ x_i \}$.

\begin{thm}
  \label{thm:infotdes}
  The informational power  $W(\pi_y)$ of any $d$-dimensional
  $t$-design POVM $\pi_y$ satisfies
  \begin{align*}
    W(\pi_y)   \le    \ln   d   -   d    \sum_{k=1}^t   a_k
    \binom{d+k-1}{k}^{-1},
  \end{align*}
  where $a_k$  are the coefficients of  the polynomial $p(x)
  := \sum_{k=1}^t  a_k x^k$ such that  $p(1) = 0$ if  $t$ is
  odd and
  \begin{align*}
    p^{(j)}(x_i) & = \eta^{(j)}(x_i),  \quad & \forall 1 \le
    i \le \floor{t/2}, \; j = 0,1,
  \end{align*}
  for some  choice of  $\{ x_i  \}_{i=1}^{\floor{t/2}}$ such
  that $0 < x_i < 1$ and $x_i < x_{i'}$ for any $i < i'$.
\end{thm}

\begin{proof}
  By direct inspection one has
  \begin{align*}
    \eta^{(j)}(x)  =  (-)^{j-1}   (j-2)!   x^{-j+1},  \qquad
    \forall j \ge 2,
  \end{align*}
  then $\eta^{(j)}(x) < 0$ for even $j$ and $\eta^{(j)}(x) >
  0$ for odd $j$,  for any $j > 1$ and $x  \in [0, 1]$. Then
  by  Lemma~\ref{thm:strategy}   one  has  that   $p(x)  \le
  \eta(x)$ for $x \in [0,1]$, and by Theorem~\ref{thm:bound}
  one has
  \begin{align*}
    W(\pi_y) \le  & \ln d -  \inf_{\psi_x} \sum_{k=1}^t a_k
    \iint \frac {|\bra{\psi_x} \pi_y \ket{\psi_x}|^k} { ( ||
      \psi_x||^2 \Tr[ \pi_y ] )^{k-1} } dx dy.
  \end{align*}
  By           Definition~\ref{def:coincidence}          and
  Lemma~\ref{thm:coincidence} one has
  \begin{align*}
    W(\pi_y)   \le    \ln   d   -   d    \sum_{k=1}^t   a_k
    \binom{d+k-1}{k}^{-1} \inf_{\psi_x} \int  || \psi_x ||^2
    dx,
  \end{align*}
  so the statement immediately follows.
\end{proof}

\begin{rmk}
  Since the $a_k$ depend upon the choice of $\{ x_i \}$, the
  tightest bound provided by Theorem~\ref{thm:infotdes} is
  \begin{align}
    \label{eq:optimal}
    W(\pi_y) \le \ln  d - d \sup_{ \{  x_i \} }\sum_{k=1}^t
    a_k \binom{d+k-1}{n}^{-1}.
  \end{align}
\end{rmk}

\subsection{Applications}
\label{sec:application}

In  this subsection  we  solve the  optimization problem  in
Eq.~\eqref{eq:optimal}  to   derive  upper  bounds   on  the
informational  power of  $t$-designs  as a  function of  the
dimension  $d$, for  $t \in  [1,5]$  and $t  = \infty$,  and
asymptotic formulae  for $d \to  \infty$.  The case $t  = 1$
coincides with the well-known Holevo~\cite{Hol73} bound; the
case $t =  2$ was already derived  in Ref.~\cite{Dal14}; the
case $t  = \infty$ coincides with  the well-known subentropy
bound~\cite{JRW94}.

\begin{cor}[Informational power of $1$-designs]
  \label{thm:info1des}
  For any  $1$-design POVM $\pi_y$, the  informational power
  $W(\pi_y)$ is upper bounded by $W(\pi_y) \le W_1(d)$, with
  $W_1(d) := \ln d$.
\end{cor}

\begin{proof}
  There  is  actually no  optimization  in  this case  since
  $\floor{t/2}    =    0$    so     the    set    $\{    x_i
  \}_{i=1}^{\floor{t/2}}$ is empty. The statement follows by
  direct inspection.
\end{proof}

\begin{cor}[Informational power of $2$-designs]
  \label{thm:info2des}
  For any  $2$-design POVM $\pi_y$, the  informational power
  $W(\pi_y)$ is upper bounded by $W(\pi_y) \le W_2(d)$, with
  \begin{align*}
    W_2(d) := \ln\frac{2d}{d+1}.
  \end{align*}
\end{cor}

\begin{proof}
  The supremum in Eq.~\eqref{eq:optimal} is achieved by $x_1
  =  2/(d+1)$.    Then  the  statement  follows   by  direct
  inspection.
\end{proof}

\begin{rmk}
  \label{rmk:info2des}
  The  limit  for $d  \to  \infty$  of  the upper  bound  in
  Corollary~\ref{thm:info2des} is given by
  \begin{align*}
    W_2(d)  \to  \ln  2   =  1  \textrm{bit}  \simeq  0.693
    \textrm{nat}
  \end{align*}
\end{rmk}

\begin{cor}[Informational power of $3$-designs]
  \label{thm:info3des}
  For any  $3$-design POVM $\pi_y$, the  informational power
  $W(\pi_y)$ is upper bounded by $W(\pi_y) \le W_3(d)$, with
  \begin{align*}
    W_3(d)       :=       \ln\frac{2d}{d+2}       +       2
    \frac{\ln\frac{d+2}2}{d(d+1)}.
  \end{align*}
\end{cor}

\begin{proof}
  The supremum in Eq.~\eqref{eq:optimal} is achieved by $x_1
  =  2/(d+2)$.    Then  the  statement  follows   by  direct
  inspection.
\end{proof}

\begin{rmk}
  \label{rmk:info3des}
  The  limit  for $d  \to  \infty$  of  the upper  bound  in
  Corollary~\ref{thm:info3des} is given by
  \begin{align*}
    W_3(d)  \to  \ln  2   =  1  \textrm{bit}  \simeq  0.693
    \textrm{nat}
  \end{align*}
\end{rmk}

\begin{cor}[Informational power of $4$-designs]
  \label{thm:info4des}
  For any  $4$-design POVM $\pi_y$, the  informational power
  $W(\pi_y)$ is upper bounded by $W(\pi_y) \le W_4(d)$, with
  \begin{widetext}
    \begin{align*}
      W_4(d) := \frac12 \ln \frac{6d^2}{(d+2)(d+3)} + \frac
      {(d-3)\sqrt{3d(d+2)}}     {6d(d+1)}     \ln     \frac
      {2d+3-\sqrt{3d(d+2)}} {d+3},
    \end{align*}
  \end{widetext}
\end{cor}

\begin{proof}
  The supremum in Eq.~\eqref{eq:optimal} is achieved by
  \begin{align*}
    x_{1,2} = \frac {3d +6 \pm \sqrt{3d(d+2)}} {d^2+5d+6}
  \end{align*}
  Then the statement follows by direct inspection.
\end{proof}

\begin{rmk}
  \label{rmk:info4des}
  The  limit  for $d  \to  \infty$  of  the upper  bound  in
  Corollary~\ref{thm:info4des} is given by
  \begin{align*}
    W_4(d)          \to           \frac{\ln6}{2}          +
    \frac{\ln(2-\sqrt3)}{2\sqrt{3}}       \simeq      0.744
    \textrm{bit} \simeq 0.516 \textrm{nat}
  \end{align*}
\end{rmk}

\begin{cor}[Informational power of $5$-designs]
  \label{thm:info5des}
  For any  $5$-design POVM $\pi_y$, the  informational power
  $W(\pi_y)$ is upper bounded by $W(\pi_y) \le W_5(d)$, with
  \begin{widetext}
    \begin{align*}
      W_5(d)  :=  &  \ln  d  +  \frac{(d-1)(d+3)(d^2+2d+4)}
      {2d(d+1)^2(d+2)}  \ln  \frac{6}{(d+3)(d+4)}  \\  &  +
      \frac                     {\sqrt{d+3}(d-1)(d^2-2d-12)}
            {2\sqrt3d(d+1)^{\frac{3}{2}}(d+2)}   \ln  \frac
            {2d+5-\sqrt{3(d+1)(d+3)}} {d+4}
    \end{align*}
  \end{widetext}
\end{cor}

\begin{proof}
  The supremum in Eq.~\eqref{eq:optimal} is achieved by
  \begin{align*}
    x_{1,2} = \frac{3d+9 \pm \sqrt{3(d^2+4d+3)}}{d^2+7d+12}
  \end{align*}
  Then the statement follows by direct inspection.
\end{proof}

\begin{rmk}
  \label{rmk:info5des}
  The  limit  for $d  \to  \infty$  of  the upper  bound  in
  Corollary~\ref{thm:info5des} is given by
  \begin{align*}
    W_5(d)          \to           \frac{\ln6}{2}          +
    \frac{\ln(2-\sqrt3)}{2\sqrt{3}}       \simeq      0.744
    \textrm{bit} \simeq 0.516 \textrm{nat}
  \end{align*}
\end{rmk}

\begin{cor}
  \label{thm:infoinfdes}
  For  the continuous  $d$-dimensional $\infty$-design  POVM
  $\pi_y$ the informational power $W(\pi_y)$ is given by
  \begin{align*}
    W(\pi_y) = W_\infty(d) := \ln d - \sum_{n=2}^d n^{-1}.
  \end{align*}
\end{cor}

\begin{proof}
  By  expanding $\eta(x)$  in Taylor  series around  $1$ and
  applying the binomial theorem one has
  \begin{align*}
    \eta(x)  =  1  -   x  -  \sum_{n=2}^\infty  \sum_{k=0}^n
    \frac{(n-2)!}{k!(n-k)!}(-x)^{k}.
  \end{align*}
  Then         by        Theorem~\ref{thm:bound}         and
  Lemma~\ref{thm:coincidence} one has
  \begin{align*}
    & W(\pi_y) \\ \le & \ln d - d + 1 + d \sum_{n=2}^\infty
    \sum_{k=0}^n              \frac{(n-2)!(-)^{k}}{k!(n-k)!}
    \binom{d-1+k}{k}^{-1}.
  \end{align*}
  Then  by direct  inspection (see  e.g. Ref.~\cite{OLBC10})
  one has
  \begin{align*}
    W(\pi_y) \le \ln d - \sum_{n=2}^d n^{-1}.
  \end{align*}
  Since  this   bound  is   saturated  by   any  orthonormal
  ensemble~\cite{JRW94}, the statement follows.
\end{proof}

\begin{rmk}
  The  limit   for  $d  \to  \infty$   of  $W_\infty(d)$  in
  Corollary~\ref{thm:infoinfdes} is given by
  \begin{align*}
    W_\infty(d)  \to 1  - \gamma  \simeq 0.610  \textrm{bit}
    \simeq 0.423 \textrm{nat},
  \end{align*}
  where $\gamma$ represents the Euler-Mascheroni constant.
\end{rmk}

We  conclude  this  subsection by  summarizing  the  derived
bounds  in Table~\ref{tab:bounds}  and illustrating  them in
Fig.~\ref{fig:bounds}.

\begin{widetext}
  \begin{center}
    \begin{table}[h]
      \begin{tabular} {| >{$}m{.025\textwidth}<{$} | >{$}m{.8\textwidth}<{$} | >{$}m{.15\textwidth}<{$} |}
        \hline {\bf  t} &  {\bf W_t(d)}  & {\bf  \lim_{d \to
            \infty}} \\ \hline 1 & \ln d & \infty \\ \hline
        2  &  \ln\frac{2d}{d+1} &  \ln  2  \\ \hline  3  &
        \ln\frac{2d}{d+2} + 2\frac{\ln\frac{d+2}2}{d(d+1)}
        & \ln  2 \\  \hline 4 &  \frac12 \ln  \frac {6d^2}
                 {(d+2)(d+3)} +  \frac {(d-3)\sqrt{3d(d+2)}}
                 {6d(d+1)} \ln  \frac {2d+3-\sqrt{3d(d+2)}}
                 {d+3}       &       \frac{\ln6}{2}       +
                 \frac{\ln(2-\sqrt3)}{2\sqrt{3}}  \\ \hline
                 5  & \ln(d)  + \frac{(d-1)(d+3)(d^2+2d+4)}
                 {2d(d+1)^2(d+2)}     \ln     \frac     {6}
                 {(d+3)(d+4)}             +            \frac
                 {\sqrt{d+3}(d-1)(d^2-2d-12)}
                 {2\sqrt3d(d+1)^{\frac{3}{2}}(d+2)}     \ln
                 \frac  {2d+5-\sqrt{3(d+1)(d+3)}}   {d+4}  &
                 \frac{\ln6}{2}                           +
                 \frac{\ln(2-\sqrt3)}{2\sqrt{3}}  \\ \hline
                 \infty & \ln d - \sum_{n=2}^d n^{-1} & 1 -
                 \gamma \\ \hline
      \end{tabular}
      \caption{Upper  bounds $W_t(d)$  on the  informational
        power $W(\pi_y)$  of any  $d$-dimensional $t$-design
        POVM $\pi_y$  for $t  \in [1,5]$  and $t  = \infty$,
        along with their asymptotic formulae.}
      \label{tab:bounds}
    \end{table}
  \end{center}
\end{widetext}

\begin{figure}[h]
  \includegraphics[width=\columnwidth]{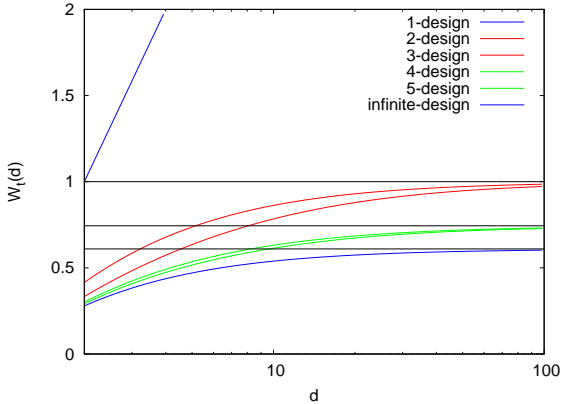}
  \caption{(Color  online)  Upper  bounds  $W_t(d)$  on  the
    informational power $W(\pi_y)$ (in  bits) of any quantum
    $t$-design POVM  $\pi_y$ as a function  of the dimension
    $d$ (on a log scale). From  top to bottom: $t = 1$ (blue
    line), $t  = 2$  and $3$  (red lines), $t  = 4$  and $5$
    (green lines), and $t = \infty$ (blue line), as provided
    by  Corollaries~\ref{thm:info1des},  \ref{thm:info2des},
    \ref{thm:info3des},                  \ref{thm:info4des},
    \ref{thm:info5des},            and~\ref{thm:infoinfdes},
    respectively.     The    asymptotes   $W_{2,3}(d)    \to
    1\textrm{bit}$, $W_{4,5}(d)  \to 0.744\textrm{bit}$, and
    $W_\infty(d)  \to 0.609  \textrm{bit}$ are  depicted too
    (horizontal black lines).}
 \label{fig:bounds}
\end{figure}

\subsection{Tightness}
\label{sec:tightness}

The bound  in Theorem~\ref{thm:infotdes} is of  course tight
for $t =  1$ for any dimension $d$,  where optimal ensembles
are  given by  any orthonormal  basis~\cite{Hol73}. In  this
subsection  we   prove  tightness  for   $2,3,5$-designs  in
dimension $2$, and for $2$-designs  in dimension $3$. For $d
=  2$ the  Bloch-sphere  representation  provides a  natural
isomorphism  between  $2$-dimensional  POVMs and  solids  in
$\mathbb{R}^3$, so we  will denote POVMs by the  name of the
corresponding      solid      (tetrahedron,      octahedron,
icosahedron). Formal  definitions of each POVM  can be found
in~\cite{DBO14, SS14}.

The informational power  of the $2$-dimensional tetrahedral,
octahedral,   and   icosahedral   POVMs  were   derived   in
Refs.~\cite{DDS11, SS14, DBO14,  Szy14, Dal14}.  By noticing
that   these  POVMs   are   $2$-,   $3$-  and   $5$-designs,
respectively,  their  informational power  directly  follows
from Theorem~\ref{thm:infotdes}.

\begin{cor}
  \label{thm:tetrahedral}
  The $2$-dimensional  tetrahedral (SIC)  POVM $\pi_y$  is a
  $2$-design, its informational power is given by
  \begin{align*}
    W_2(2) = \ln\frac43,
  \end{align*}
  and  the optimal  (anti-tetrahedral) ensemble  $\psi_x$ is
  such that $\psi_x \pi_x = 0$ for any $x$.
\end{cor}

\begin{proof}
  Any  SIC  POVM  is  a  $2$-design~\cite{RBSC04},  and  the
  anti-tetrahedral   ensemble   saturates   the   bound   in
  Corollary~\ref{thm:info2des}.
\end{proof}

\begin{cor}
  \label{thm:octahedral}
  The $2$-dimensional  octahedral (complete  MUB) POVM  is a
  $3$-design, its informational power is given by
  \begin{align*}
    W_3(2) = \frac16 \ln4,
  \end{align*}
  and  the optimal  (anti-octahedral)  ensemble $\psi_x$  is
  such that $\psi_x \pi_x = 0$ for any $x$.
\end{cor}

\begin{proof}
  It follows  by direct inspection that  the $2$-dimensional
  octahedral POVM  is a $3$-design, and  the anti-octahedral
  ensemble        saturates        the       bound        in
  Corollary~\ref{thm:info3des}.
\end{proof}

\begin{cor}
  \label{thm:icosahedral}
  The $2$-dimensional icosahedral POVM  is a $5$-design, its
  informational power is given by
  \begin{align*}
    W_5(2) =  \ln2 - \frac5{12}\ln5  - \frac{\sqrt{5}}{12}
    \ln\frac{9-3\sqrt5}6
  \end{align*}
  and  the optimal  (anti-icosahedral) ensemble  $\psi_x$ is
  such that $\psi_x \pi_x = 0$ for any $x$.
\end{cor}

\begin{proof}
  It follows  by direct inspection that  the $2$-dimensional
  icosahedral POVM is a $5$-design, and the anti-icosahedral
  ensemble        saturates        the       bound        in
  Corollary~\ref{thm:info5des}.
\end{proof}

In  Ref.~\cite{Szy14} it  was shown  that the  informational
power of group covariant  $3$-dimensional SIC POVMs is given
by $W_2(3) = \ln\frac32$.  By noticing that these POVMs are
$2$-designs,  the  optimality   of  this  value  immediately
follows from Corollary~\ref{thm:info2des}.

\section{Conclusion and outlook}
\label{sec:conclusion}

In this  work we  provided in  Theorem~\ref{thm:infotdes} an
upper bound  on the information  that can be carried  by any
quantum  $t$-design  for  any  $t$, as  a  function  of  the
dimension       of       the      system,       and       in
Corollaries~\ref{thm:info1des},          \ref{thm:info2des},
\ref{thm:info3des},  \ref{thm:info4des}, \ref{thm:info5des},
and~\ref{thm:infoinfdes}   we    derived   closed   analytic
expressions  for such  bound for  $t \in  [1, 5]$  and $t  =
\infty$.   The  Holevo   upper  bound~\cite{Hol73}  and  the
subentropy  lower   bound~\cite{JRW94}  were   recovered  as
particular cases for $t = 1$ and $t = \infty$, respectively.
In this sense, the  resulting hierarchy of bounds represents
a trade-off between  the uniformity of a  quantum system and
the amount information it  can carry. By deriving asymptotic
formulae  for  large dimensions,  we  also  showed that  the
statistics generated by any $t$-design contains no more than
a single bit of information,  and that this amount decreases
with            $t$.              Furthermore,            in
Corollaries~\ref{thm:tetrahedral},     \ref{thm:octahedral},
and~\ref{thm:icosahedral}  we showed  the  tightness of  our
bounds  for  qubits  and  qutrits.   Finally,  as  a  direct
consequence  of   Theorem~\ref{thm:duality}  it  immediately
follows  that   all  the  presented  upper   bounds  on  the
informational  power  of  $t$-design POVMs  holds  as  upper
bounds   on  the   accessible   information  of   $t$-design
ensembles.

Various open problems related  to the accessible information
and   informational  power   of  quantum   $t$-designs  were
discussed  in  Refs.~\cite{DBO14,  Dal14}. In  view  of  the
results presented here, we may add to the list the following
questions.  The  asymptotic formulae  for the bounds  on the
informational    power   of    $2$-    and   $3$-    designs
(Remarks~\ref{rmk:info2des} and~\ref{rmk:info3des}), as well
as      those       for      $4$-       and      $5$-designs
(Remarks~\ref{rmk:info4des}   and~\ref{rmk:info5des}),   are
pairwise identical.  Can this  be generalized to higher $t$?
Can  this phenomenon  be  given  a physical  interpretation?
Moreover,  for all  the  $t$-design qubit  POVMs $\pi_y$  we
explicitly   optimized   (Corollaries~\ref{thm:tetrahedral},
\ref{thm:octahedral},     and~\ref{thm:icosahedral}),    the
optimal ensemble $\psi_x$ turned out to be such that $\psi_x
\pi_x =  0$ for any  $x$.  Is it  always the case  for qubit
$t$-designs?  Finally,  closed analytic expressions  for the
bounds provided by Theorem~\ref{thm:infotdes}  for $t \ge 6$
require lengthy calculations, therefore their derivation can
be made easier by the use of a symbolic calculation package.
This will be done  in a forthcoming work~\cite{Dal15}, where
their  tightness  and  asymptotic   formulae  will  also  be
discussed.

During the  preparation of  this manuscript, the  author was
informed   by    Wojciech   S{\l}omczy\'{n}ski    and   Anna
Szymusiak~\cite{SS15} of  a recent result of  theirs showing
that the bound  in Corollary~\ref{thm:info2des} is saturated
by the $64$ Hoggar lines SIC-POVM in dimension $8$.

\section*{Acknowledgments}

The author is indebted to Francesco Buscemi and Massimiliano
F.  Sacchi  for their  insightful comments  and suggestions,
and  to Anna  Szymusiak  for  stimulating discussions  about
Hermite  polynomial  interpolation.    The  author  is  also
grateful to Alessandro Bisio,  Chris Fuchs, Paolo Perinotti,
Wojciech  S{\l}omczy\'{n}ski,  and  Vlatko Vedral  for  very
useful discussions.  This work was supported by the Ministry
of Education and the Ministry of Manpower (Singapore).

\end{document}